\documentclass[11pt]{article}
\usepackage[utf8]{inputenc}
\usepackage{amsfonts,latexsym,amsthm,amssymb,amsmath,amscd,euscript,mathrsfs,graphicx}
\usepackage{framed}
\usepackage{fullpage}
\usepackage{hyperref}
    \hypersetup{colorlinks=true,citecolor=blue,urlcolor =black,linkbordercolor={1 0 0}}
\usepackage{enumitem}
\usepackage[usenames]{color}

\allowdisplaybreaks[1]

\newtheorem{theorem}{Theorem}[section]
\newtheorem{proposition}{Proposition}[section]
\newtheorem{lemma}{Lemma}[section]
\newtheorem{corollary}{Corollary}[section]

\theoremstyle{definition}
\newtheorem{defn}{Definition}[section]

\newcommand{\BE}{\mathbb{E}}

\newcommand{\BP}{\mathbb{P}}
\newcommand{\BR}{\mathbb{R}}

\newcommand{\EFN}{\mathrm{EFN}}
\newcommand{\MMMR}{\mathrm{MMMR}}
\newcommand{\Ext}{\mathrm{Ext}}

\theoremstyle{remark}
\newtheorem*{remark}{Remark}

\title{Optimal terminal dimensionality reduction in Euclidean space}
\author{Shyam Narayanan\thanks{Harvard University. \texttt{shyamnarayanan@college.harvard.edu}. Supported by a PRISE fellowship and a Herchel-Smith Fellowship.}
\and
Jelani Nelson\thanks{Harvard University. \texttt{minilek@seas.harvard.edu}. Supported by NSF
  CAREER award CCF-1350670, NSF grant IIS-1447471, ONR grant N00014-18-1-2562, ONR DORECG award N00014-17-1-2127, an Alfred P.\ Sloan Research Fellowship, and a Google Faculty Research Award.}}

\begin{document}

\maketitle

\begin{abstract}
Let $\varepsilon\in(0,1)$ and $X\subset\BR^d$ be arbitrary with $|X|$ having size $n>1$. The Johnson-Lindenstrauss lemma states there exists $f:X\rightarrow\BR^m$ with $m = O(\varepsilon^{-2}\log n)$ such that
\begin{equation*}
\forall x\in X\ \forall y\in X, \|x-y\|_2 \le \|f(x)-f(y)\|_2 \le (1+\varepsilon)\|x-y\|_2 .
\end{equation*}
We show that a strictly stronger version of this statement holds, answering one of the main open questions of \cite{MahabadiMMR18}: ``$\forall y\in X$'' in the above statement may be replaced with ``$\forall y\in\BR^d$'', so that $f$ not only preserves distances within $X$, but also distances {\it to} $X$ from the rest of space. Previously this stronger version was only known with the worse bound $m = O(\varepsilon^{-4}\log n)$. Our proof is via a tighter analysis of (a specific instantiation of) the embedding recipe of \cite{MahabadiMMR18}.
\end{abstract}

\section{Introduction}\label{sec:intro}
{\it Metric embeddings} may play a role in algorithm design when input data is geometric, in which the technique is applied as a pre-processing step to map input data living in some metric space $(\mathcal X, d_{\mathcal X})$ into some target space $(\mathcal Y, d_{\mathcal Y})$ that is algorithmically friendlier. One common approach is that of {\it dimensionality reduction}, in which $\mathcal X$ and $\mathcal Y$ are both subspaces of the same normed space, but where $\mathcal Y$ is of much lower dimension. Working with lower-dimensional embedded data then typically results in efficiency gains, in terms of memory, running time, and/or other resources.

A cornerstone result in this area is the Johnson-Lindenstrauss (JL) lemma \cite{JohnsonL84}, which provides dimensionality reduction for Euclidean space.
\begin{lemma}{\cite{JohnsonL84}}
Let $\varepsilon\in(0,1)$ and $X\subset\BR^d$ be arbitrary with $|X|$ having size $n>1$. There exists $f:X\rightarrow\BR^m$ with $m = O(\varepsilon^{-2}\log n)$ such that
\begin{equation}
\forall x,y\in  X, \|x-y\|_2 \le \|f(x)-f(y)\|_2 \le (1+\varepsilon)\|x-y\|_2 . \label{eqn:jl2}
\end{equation}
\end{lemma}
It has been recently shown that the dimension $m$ of the target Euclidean space achieved by the JL lemma is best possible, at least for $\varepsilon \gg 1/\sqrt{\min(n,d)}$ \cite{LarsenN17} (see also \cite{AlonK17}).

The multiplicative factor on the right hand side of Eqn.~\eqref{eqn:jl2}, in this case $1+\varepsilon$, is referred to as the {\em distortion} of the embedding $f$. Recent work of Elkin et al.\ \cite{ElkinFN17} showed a stronger form of Euclidean dimensionality reduction for the case of constant distortion. Namely, they showed that in Eqn.~\eqref{eqn:jl2}, whereas $x$ is taken as an arbitrary point in $X$, $y$ may be taken as an arbitrary point in $\BR^d$. They called such an embedding a {\it terminal embedding}\footnote{More generally, a terminal embedding from $(\mathcal X, d_{\mathcal X})$ into $(\mathcal Y, d_{\mathcal Y})$ with terminal set $K\subset \mathcal X$ and {\it terminal distortion} $\alpha$ is a map $f:\mathcal X\rightarrow\mathcal Y$ s.t.\ $\exists\ c>0$ satisfying $d_{\mathcal X}(u,w) \le c\cdot d_{\mathcal Y}(f(u), f(w)) \le \alpha\cdot d_{\mathcal X}(u, w)$ for all $u\in K, w \in \mathcal X$.}. Though rather than achieving terminal distortion $1+\varepsilon$, their work only showed how to achieve constant terminal distortion with $m = O(\log n)$ for a constant that could be made arbitrarily close to $\sqrt{10}$ (see \cite[Theorem 1]{ElkinFN17}). Terminal embeddings can be useful in static high-dimensional computational geometry data structural problems. For example, consider nearest neighbor search over some finite database $X\subset\BR^d$. If one builds a data structure over $f(X)$ for a terminal embedding $f$, then {\it any} future query is guaranteed to be handled correctly (or, at least, the embedding will not be the source of failure). Contrast this with the typical approach where one uses a randomized embedding oblivious to the input (e.g.\ random projections) that preserves the distance between any fixed pair of vectors with probability $1-1/\mathop{poly}(n)$. One can verify that the embedding preserves distances {\it within} $X$ during pre-processing, but for any later query there is some non-zero probability that the embedding will fail to preserve the distance between the query point $q$ and some points in $X$.

Subsequent to \cite{ElkinFN17}, work of Mahabadi et al.\ \cite{MahabadiMMR18} gave a construction for terminal dimensionality reduction in Euclidean space achieving terminal distortion $1+\varepsilon$, with $m = O(\varepsilon^{-4}\log n)$. They asked as one of their main open questions (see \cite[Open Problem 3]{MahabadiMMR18}) whether it is possible to achieve this terminal embedding guarantee with $m = O(\varepsilon^{-2}\log n)$, which would be optimal given the JL lower bound of \cite{LarsenN17}. Our contribution in this work is to resolve this question affirmatively; the following is our main theorem.

\begin{theorem}
Let $\varepsilon\in(0,1)$ and $X\subset\BR^d$ be arbitrary with $|X|$ having size $n>1$. There exists $f:X\rightarrow\BR^m$ with $m = O(\varepsilon^{-2}\log n)$ such that
\begin{equation*}
\forall x\in X,\forall y\in \BR^d,\ \|x-y\|_2 \le \|f(x)-f(y)\|_2 \le (1+\varepsilon)\|x-y\|_2 . 
\end{equation*}
\end{theorem}

The embedding we analyze in this work is in fact one that fits within the family introduced in \cite{MahabadiMMR18}; our contribution is to provide a sharper analysis.

We note that unlike in \cite{MahabadiMMR18}, which provided a deterministic polynomial time construction of the terminal embedding, we here only provide a Monte Carlo polynomial time construction algorithm. Our error probability though only comes from mapping the points in $X$. If the points in $X$ are mapped to $\BR^m$ well (with low ``convex hull distortion'', which we define later), which happens with high probability, then our final terminal embedding is guaranteed to have low terminal distortion as map from all of $\BR^d$ to $\BR^m$.

\bigskip

\begin{remark}\label{rem:nonlinear}
Unlike in the JL lemma for which the embedding $f$ may be linear, the terminal embedding we analyze here (as was also in the case in \cite{ElkinFN17,MahabadiMMR18}) is nonlinear, as it must be. To see that it must be nonlinear, consider that any linear embedding with constant terminal distortion, $x\mapsto \Pi x$ for some matrix $\Pi\in\BR^{m\times d}$, must have $\Pi (x-y) \neq 0$ for any $x\in X$ and any $y\in\BR^d$ on the unit sphere centered at $x$. In other words, one needs $\mathop{ker}(\Pi)\neq \emptyset$, which is impossible unless $m\ge d$.
\end{remark}

\subsection{Overview of approach}\label{sec:approach}
We outline both the approach of previous works as well as our own. In all these approaches, one starts with an embedding $f:X\rightarrow\ell_2^m$ with good distortion, then defines an {\it outer extension} $f_\Ext$ as introduced in \cite{ElkinFN17} and defined explicitly in \cite{MahabadiMMR18}.

\begin{defn}
For $f:X\rightarrow \BR^m$ and $Z\supsetneq X$, we say $g:Z\rightarrow\BR^{m'}$ is an {\it outer extension} of $f$ if $m'\ge m$, and $g(x)$ for $x\in X$ has its first $m$ entries equal to $f(x)$ and last $m'-m$ entries all $0$.
\end{defn}

In \cite{ElkinFN17,MahabadiMMR18} and the current work, a terminal embedding is obtained by, for each $u\in \BR^d\backslash X$, defining an outer extension $f^{(u)}_\Ext$ for $Z = X\cup \{u\}$ with $m' = m+1$. Since $f^{(u)}_\Ext$ and $f^{(u')}_\Ext$ act identically on $X$ for any $u,u'\in \BR^d$, we can then define our final terminal embedding by $\tilde f(u) = (f(u), 0)$ for $u\in X$ and $\tilde f(u) = f^{(u)}_\Ext(u)$ for $u\in \BR^d\backslash X$, which will have terminal distortion $\sup_{u\in \BR^d\backslash X} \mathsf{Dist}(f^{(u)}_\Ext)$, where $\mathsf{Dist}(g)$ denotes the distortion of $g$. The main task is thus creating an outer extension with low distortion for a set $Z = X \cup \{u\}$ for some $u$. In all the cases that follow, we will have $f_\Ext(x) = (f(x),0)$ for $x\in X$, and we will then specify how to embed points $u\notin X$.

The construction of Elkin et al.\ \cite{ElkinFN17}, the ``EFN extension'', is as follows. Suppose $f:X\rightarrow\ell_p^m$ is an $\alpha$-distortion embedding. The EFN extension $f_\EFN:X\cup\{u\}\rightarrow\ell_p^{m+1}$ is then defined by $f_\EFN(u) = (f(\rho_{\ell_p}(u)), d(\rho_{\ell_p}(u), u))$, where $\rho_d(u) = \mathop{argmin}_{x\in X} d(x, u)$ for some metric $d$. In the case that we view $X\cup\{u\}$ as living in the metric space $\ell_p^d$, Elkin et al.\ showed that the EFN extension has terminal distortion at most $2^{(p-1)/p} \cdot ((2\alpha)^p + 1)^{1/p}$ \cite[Theorem 1]{ElkinFN17}. If $p=2$ and $f$ is obtained via the JL lemma to have distortion $\alpha\le 1+\varepsilon$, this implies the EFN extension would have terminal distortion at most $\sqrt{10} + O(\varepsilon)$.

The EFN extension does not in general achieve terminal distortion $1+\varepsilon$, even if starting with $f$ a perfect isometry. In fact, the bound of \cite{ElkinFN17} showing distortion at least $\sqrt{10}$ is sharp. Consider for example $X = \{-1,0,2\}\subset \BR$. Consider the identity map $f(x) = x$, which has distortion $1$ as a map from $(X, \ell_2)$ to $(\BR, \ell_2)$. Then $f_\EFN(-1) = (-1,0)$, $f_\EFN(0) = (0,0)$, $f_\EFN(2) = (2,0)$, and $f_\EFN(1) = (0, 1)$, and thus $\|f_\EFN(1) - f_\EFN(-1)\|_2 = \sqrt 2$ (distance shrunk by a $\sqrt 2$ factor) and $\|f_\EFN(1) - f_\EFN(2)\|_2 = \sqrt 5$ (distance increased by a $\sqrt 5$ factor). Thus $f_\EFN$ has terminal distortion at least (in fact exactly equal to) $\sqrt{10}$. This example in fact shows sharpness for $\ell_p$ for all $p\ge 1$.

Thus to achieve terminal distortion $1+\varepsilon$, the work of \cite{MahabadiMMR18} had to develop a new outer extension, the ``MMMR extension'', which they based on the following core lemma.
\begin{lemma}[{\cite[Lemma 3.1 (rephrased)]{MahabadiMMR18}}]\label{lem:mmmr}
Let $X$ be a finite subset of $\ell_2^d$, and suppose $f:X\rightarrow\ell_2^m$ has distortion $1+\gamma$. Fix some $u\in \BR^d$, and define $x_0 := \rho_{\ell_2}(u)$. Then $\exists u'\in\BR^m$ s.t.
\begin{itemize}
\item $\|u' - f(x_0)\|_2 \le \|u - x_0\|_2$, and
\item $\forall x\in X$, $|\langle u' - f(x_0), f(x) - f(x_0)\rangle - \langle u - x_0, x - x_0\rangle| \le 3\sqrt{\gamma}(\|x - x_0\|_2^2 + \|u - x_0\|_2^2)$  
\end{itemize}
\end{lemma}
Mahabadi et al.\ then used the $u'$ promised by Lemma~\ref{lem:mmmr} as part of a construction that takes a $(1+\gamma)$-distortion embedding $f:X\rightarrow \ell_2^m$ and uses it in a black box way to construct an outer extension $f_\MMMR:X\cup\{u\}\rightarrow\ell_2^{m+1}$. In particular, they define $f_\MMMR(u) = (u', \sqrt{\|u - x_0\|_2^2 - \|u' - f(x_0)\|_2^2})$. It is clear this map perfectly preserves the distance from $u$ to $x_0$; in \cite[Theorem 1.5]{MahabadiMMR18}, it is furthermore shown that $f_\MMMR$ preserves distances from $u$ to all of $X$ up to a $1+O(\sqrt{\gamma})$ factor. Thus one should set $\gamma = \Theta(\varepsilon^2)$ so that $f_\MMMR$ has distortion $1+\varepsilon$, which is achieved by starting with an $f$ guaranteed by the JL lemma with $m = \Theta(\varepsilon^{-4}\log n)$.

They then showed that this loss is {\it tight}, in the sense that there exist $X$, $u$, $f:X\rightarrow\BR^m$, where $f$ has distortion $1+\gamma$, such that {\it any} outer extension $f_\Ext$ to domain $X\cup\{u\}$ has distortion $1+\Omega(\sqrt{\gamma})$ \cite[Section 3.2]{MahabadiMMR18}. Thus, seemingly a new approach is needed to achieve $m = O(\varepsilon^{-2}\log n)$.

One may be discouraged by the above-mentioned tightness of the $\gamma\rightarrow\Omega(\sqrt{\gamma})$ loss, but in this work we show that, in fact, the MMMR extension can be made to provide $1+\varepsilon$ distortion with the optimal $m = O(\varepsilon^{-2}\log n)$! The tightness result mentioned in the last paragraph is only an obstacle to using the low-distortion property of $f$ {\it in a black box way}, as it is only shown that there {\it exist} $f$ where the $\gamma\rightarrow\Omega(\sqrt{\gamma})$ loss is necessary. However, the $f$ we are using is not an arbitrary $f$, but rather is the $f$ obtained via the JL lemma. A standard way of proving the JL lemma is to choose $\Pi\in\BR^{m\times d}$ with i.i.d.\ subgaussian entries, scaled by $1/\sqrt m$ for $m = \Theta(\varepsilon^{-2}\log n)$ \cite[Exercise 5.3.3]{Vershynin18}. The low-distortion embedding is then $f(x) = \Pi x$. We show in this work that by not just using that this $f$ is a low-distortion embedding for $X$, but rather that it satisfies a stronger property we dub {\em convex hull distortion} (which we show $x\mapsto\Pi x$ does satisfy with high probability), one can achieve the desired terminal embedding result with optimal $m$.

\begin{defn}\label{def:chull}
For $T\subset S^{d-1}$ a subset of the unit sphere in $\BR^d$, and $\varepsilon\in(0,1)$, we say for $\Pi\in\BR^{m\times d}$ that $\Pi$ {\it provides $\varepsilon$-convex hull distortion for $T$} if 
$$
\forall x\in \mathrm{conv}(T),\ |\|\Pi x\|_2 - \|x\|_2 | < \varepsilon
$$
where $\mathrm{conv}(T) := \{\sum_i \lambda_i t_i : \forall i\ t_i\in T, \lambda_i \ge 0, \sum_i \lambda_i = 1\}$ denotes the convex hull of $T$.
\end{defn}

We show that a random $\Pi$ with subgaussian entries provides $\varepsilon$-convex hull distortion for $T$ with probability at least $1-\delta$ as long as $m = \Omega(\varepsilon^{-2}\log(|T|/(\varepsilon\delta)))$. We then replace Lemma~\ref{lem:mmmr} with a new lemma that shows that as long as $f(x) = \Pi x$ does not just have $1+\gamma$ distortion for $X$, but rather provides $\gamma$-convex hull distortion for $T = \{(x-y) / \|x - y\|_2 : x,y\in X\}$, then the $3\sqrt{\gamma}$ term on the RHS of the second bullet of Lemma~\ref{lem:mmmr} can be replaced with $O(\gamma)$ (plus one other technical improvement; see Lemma~\ref{ReductionLemma1}). Note $|T| = \binom n2$, so $\log|T| = O(\log n)$. We then show how to use the modified lemma to achieve an outer extension with $m = O(\varepsilon^{-2}\log(n/\varepsilon)) = O(\varepsilon^{-2}\log n)$. This last equality holds since we may assume $\varepsilon = \Omega(1/\sqrt n)$, since otherwise there is a trivial terminal embedding with $m = n = O(\varepsilon^{-2})$ with no distortion: if $d \le n$, take the identity map. Else, translate $X$ so $0\in X$; then $E:=\mathop{span}(X)$ has $\mathop{dim}(E) \le  n-1$. By rotation, we can assume $E=\mathop{span}\{e_1,\ldots,e_{n-1}\}$ so that every $x\in E$ can be written as $\sum_{i=1}^{n-1} \alpha(x)_i e_i$ for some vector $\alpha(x)\in\BR^{n-1}$. We can then define a terminal embedding $\tilde{f}:\BR^d\rightarrow\BR^n$ with $\tilde f(x) = (\alpha(\mathop{proj}_E(x)), \|\mathop{proj}_{E^\perp}(x)\|_2)$ for all $x\in\BR^d$. Here $\mathop{proj}_E$ denotes orthogonal projection onto $E$.

\section{Preliminaries}

For our optimal terminal embedding analysis, we rely on two previous results. The first result is the von Neumann Minimax theorem \cite{vonNeumann28}, which was also used in the terminal embedding analysis in \cite{MahabadiMMR18}. The theorem states the following:

\begin{theorem} \label{Minimax}
    Let $X \subset \BR^n$ and $Y \subset \BR^m$ be compact convex sets. Suppose that $f: X \times Y \to \BR$ is a continuous function that satisfies the following properties:
\begin{enumerate}
    \item $f(\cdot, y): X \to \BR$ is convex for any fixed $y \in Y$,
    \item $f(x, \cdot): Y \to \BR$ is concave for any fixed $x \in X.$
\end{enumerate}
    Then, we have that
\[\min\limits_{x \in X} \max\limits_{y \in Y} f(x, y) = \max\limits_{y \in Y} \min\limits_{x \in X} f(x, y).\]
\end{theorem}

The second result is a result of Dirksen \cite{Dirksen15, Dirksen16} that provides a uniform tail bound on subgaussian empirical processes. To explain the result, we first make the following definitions:

\begin{defn}
    A semi-metric $d$ on a space $X$ is a function $X \times X \to \BR_{\ge 0}$ such that $d(x, x) = 0,$ $d(x, y) = d(y, x),$ and $d(x, y)+d(y, z) \ge d(x, z)$ for all $x, y, z \in X$. 
\end{defn}

Note a semi-metric may have $d(x,y) = 0$ for $x\neq y$.

\begin{defn}
    Given a semi-metric $d$ on $T$, and subset $S \subset T,$ define $d(t, S) = \inf_{s \in S} d(t, s)$ for any point $t \in T$.
\end{defn}

\begin{defn}
    Given a semi-metric $d$ on $T$, define
\[\gamma_2(T, d) = \inf\limits_{\{S_r\}_{r=0}^\infty} \sup\limits_{t \in T} \sum\limits_{r \ge 0} 2^{r/2} d(t, S_r),\]
    where the first infimum runs over all subsequences $S_0 \subset S_1 \subset \ldots \subset T$ where $|S_0| = 1$ and $|S_r| \le 2^{2^r}$.
\end{defn}

\begin{defn}
    For any random variable $X$, we define the \textit{subgaussian norm} of $X$ as
\[\|X\|_{\psi_2} = \inf \left\{C \ge 0: \BE\left(e^{X^2/C^2}\right) \le 2\right\}.\]
\end{defn}

It is well known that the subgaussian norm as defined above indeed defines a norm \cite{BoucheronLM13}. We also note the following proposition:

\begin{proposition} \label{IIDSubgaussianAddition} \cite{BoucheronLM13}
    There exists some constant $k$ such that if $X_1, ..., X_n$ are i.i.d. subgaussians with mean $0$ and subgaussian norm $C$, then for any $a_1, ..., a_n \in \BR,$ $a_1X_1+...+a_nX_n$ is subgaussian with subgaussian norm $\le k \|a\|_2 C,$ where $\|a\|_2 = \sqrt{a_1^2+...+a_n^2}.$
\end{proposition}

\begin{theorem} \label{DirksenInequality} \cite[Theorem 3.2]{Dirksen16}
    Let $T$ be a set, and suppose that for every $t \in T$ and every $1 \le i \le m,$ $X_{t, i}$ is a random variable with finite expected value and variance. For each $t \in T,$ let
\[A_t = \frac{1}{m} \sum\limits_{i = 1}^{m} \left(X_{t, i}^2 - \BE X_{t, i}^2\right).\]
    Consider the following semi-metric on $T$: for $s, t \in T$, define
\[d_{\psi_2}(s, t) := \max\limits_{1 \le i \le m} \|X_{s, i} - X_{t, i}\|_{\psi_2}\]
    and define 
\[\overline{\Delta}_{\psi_2}(T) := \sup\limits_{t \in T} \max\limits_{1 \le i \le m} \|X_{t, i}\|_{\psi_2}.\]
    Then, there exist constants $c, C > 0$ such that for all $u \ge 1,$
\[\BP\left(\sup\limits_{t \in T} |A_t| \ge C \left(\frac{1}{m} \gamma_2^2 (T, d_{\psi_2}) + \frac{1}{\sqrt{m}} \overline{\Delta}_{\psi_2}(T) \gamma_2(T, d_{\psi_2})\right) + c\left(\sqrt{u} \frac{\overline{\Delta}_{\psi_2}^2(T)}{\sqrt{m}} + u \frac{\overline{\Delta}_{\psi_2}^2(T)}{m}\right)\right) \le e^{-u}.\]
\end{theorem}

\section{Construction of the Terminal Embedding}

\subsection{Universal Dimensionality Reduction with an additional $\ell_1$ condition} \label{L1JL}

Here we show that for all $X = \{x_1, \dots, x_n\}\subset\BR^d$ a set of unit norm vectors, there exists $\Pi\in\BR^{m\times d}$ for $m = O(\varepsilon^{-2} \log n)$ providing $\varepsilon$-convex hull distortion for $X$, as defined in Definition~\ref{def:chull}.

If $\varepsilon^{-2} > n,$ this construction follows by projecting onto a spanning subspace of $x_1, \dots, x_n$ and choosing an orthonormal basis.  For $\varepsilon^{-2} < n$, our goal is to show that if $\Pi \in \BR^{m \times d}$ is a normalized random matrix with i.i.d.\ subgaussian entries (normalized by $1/\sqrt m$), then $\Pi$ provides $\varepsilon$-convex hull distortion with high probability.

Define $T = \mathrm{conv}(X)$. We apply Theorem~\ref{DirksenInequality} as follows. For some $m$ which we will choose later, let $\Pi^0$ be a matrix in $\BR^{m \times d}$ with i.i.d.\ subgaussian entries with mean $0$, variance $1$, and some subgaussian norm $C_1.$  Let $\Pi$ be the scaled matrix, i.e. $\frac{1}{\sqrt{m}} \cdot \Pi^0.$  Let $\Pi_i$ denote the $i$th row of $\Pi^0$ and for any $t \in \BR^d,$ let $X_{t, i} = \Pi_i t.$  Finally, let $T_k$ be the subset of $T$ of points with norm at most $2^{-k},$ i.e., $T_k = \{x \in T: \|x\|_2 \le 2^{-k}\}$.

    First, note that 
\[A_t := \frac{1}{m} \sum\limits_{i = 1}^{m} \left((\Pi_i t)^2 - \BE(\Pi_i t)^2\right) = \frac{1}{m} \left(\| \Pi^0 t\|_2^2 - m \cdot \|t\|_2^2\right) =  \|\Pi t\|_2^2 - \|t\|_2^2.\]

    For any $t \in T$ and $1 \le i \le m,$ define $X_{t, i} = \Pi_i t$. Then, $A_t$ corresponds to the definition in Theorem~\ref{DirksenInequality}. Note that for any $s, t \in T$ and for any $1 \le i \le m,$ $X_{s, i} - X_{t, i} = X_{s-t, i},$ which is a subgaussian with mean $0$, variance $\|s-t\|_2^2,$ and subgaussian norm at most $k \|s-t\|_2 C_1$ by Proposition~\ref{IIDSubgaussianAddition}. Therefore, if we define $d_{\psi_2}(s, t) = \max_{1 \le i \le m} \|X_{s, i} - X_{t, i}\|_{\psi_2},$ as in Theorem \ref{DirksenInequality}, $d_{\psi_2}(s, t) \le k \|s-t\|_2 C_1.$  As a result, we have the following:

\begin{proposition}
    $\overline{\Delta}_{\psi_2}(T_k) = O(2^{-k})$.
\end{proposition}

\begin{proof}
    Since any point $t \in T_k$ has Euclidean norm at most $2^{-k}$, the conclusion is immediate.
\end{proof}
    
    We also note the following:

\begin{proposition}\label{prop:majorizing-measures}
    $\gamma_2(T_k, \ell_2) = O(\sqrt{\log n}),$ where $\ell_2$ represents the standard Euclidean distance.
\end{proposition}

\begin{proof}
    The Majorizing Measures theorem \cite{Talagrand14} gives $\gamma_2(T_k, \ell_2) = \Theta\left(\BE_g (\sup_{x \in T_k} \langle g, x \rangle \right)),$ where $g$ is a $d$-dimensional vector of i.i.d.\ standard normal random variables. Also, $\BE_g (\sup_{x \in T_k} \langle g, x \rangle) \le \BE_g (\sup_{x \in T} \langle g, x \rangle)$ because $T_k \subset T$. However, since $T = \mathrm{conv}(X)$ and since $\langle g, x \rangle$ is a linear function of $x$, we have $\sup_{x \in T} \langle g, x \rangle = \sup_{x\in X} \langle g, x\rangle.$  Since the $x_i$'s are unit norm vectors, then each $g_i := \langle g, x_i\rangle$ is standard normal, which implies (even if they are dependent) that $\BE \sup_i |g_i| = O(\sqrt{\log n}).$  This means that $\gamma_2(T_k, \ell_2) = O(\BE_g (\sup_{x \in T} \langle g, x \rangle)) = O(\sqrt{\log n}),$ so the proof is complete.
\end{proof}

This allows us to bound $\gamma_2(T, d_{\psi_2})$ as follows.
\begin{corollary}
    $\gamma_2(T_k, d_{\psi_2}) = O(\sqrt{\log n}).$
\end{corollary}

\begin{proof}
    As $d_{\psi_2}(s, t) = O(\|s-t\|_2)$ for any points $s, t,$ the result follows from Proposition~\ref{prop:majorizing-measures}.
\end{proof}

Therefore, using $T = T_k$ in Theorem \ref{DirksenInequality} and varying $k$ gives us the following.

\begin{theorem} \label{MainIdea}
    Suppose that $\varepsilon, \delta < 1$ and $m = \Theta\left(\frac{1}{\varepsilon^2} \log \frac{n \log (2/\varepsilon)}{\delta}\right).$ Then, there exists some constant $C_2$ such that for all $k \ge 0,$
\[\BP\left(\sup\limits_{t \in T_k}  \left|\| \Pi t\|_2^2 - \|t\|_2^2\right| \ge C_2\left(\varepsilon^2 + \varepsilon \cdot 2^{-k}\right)\right) \le \frac{\delta}{n \log (2/\varepsilon)}.\]
    Consequently, with probability at least $1 - \delta/n,$ we have that for all $t \in T$,
\[\left|\|\Pi t\|_2 - \|t\|_2\right| = O(\varepsilon).\]
\end{theorem}

\begin{proof}
    The first part follows from Theorem \ref{DirksenInequality} and the previous propositions. Note that
\[\sup\limits_{t \in T_k}  \left|\|\Pi t\|_2^2 - \|t\|_2^2\right| = \sup\limits_{t \in T_k} |A_t|.\]
    Moreover, $\frac{1}{m} \gamma_2^2(T_k, d_{\psi_2}) = O(\varepsilon^2)$ and $\frac{1}{\sqrt{m}} \overline{\Delta}_{\psi_2}(T_k) \gamma_2(T_k, d_{\psi_2}) = O(\varepsilon \cdot 2^{-k}).$ If we let $u = \ln \frac{n \log (2/\varepsilon)}{\delta},$ then $\sqrt{u} \frac{\overline{\Delta}_{\psi_2}^2(T_k)}{\sqrt{m}} = O(\varepsilon \cdot 2^{-2k}),$ and $u \frac{\overline{\Delta}_{\psi_2}^2(T)}{m} = O(\varepsilon^2 \cdot 2^{-2k}).$  The first result now follows immediately from Theorem \ref{DirksenInequality}, as $e^{-u} = \frac{\delta}{n \log (2/\varepsilon)}.$ Note that it is possible for $T_k$ to be empty, but in this case we see that $\sup_{t \in T_k} |A_t| = 0$ so the first result is immediate.
    
    For the second part, assume WLOG that $\varepsilon^{-1} = 2^\ell$ for some $\ell.$  Define $T_0', T_1', \dots, T_{\ell}'$ such that $T_{\ell}' = T_{\ell}$ and for all $k < \ell,$ $T_k' = T_k \backslash T_{k+1}.$  Then, $T_0', \dots, T_{\ell}'$ forms a partition of $T$, since $\|x\|_2 \le 1$ if $x \in T$. Note that if $T \in T_{\ell},$ then with probability at least $1 - \frac{\delta}{n\log(2/\varepsilon)}$, $\left|\|\Pi t\|_2^2 - \|t\|_2^2\right| = O(\varepsilon^{2})$ and thus $\|\Pi t\|_2^2 = O(\varepsilon^2)$ and $\left|\|\Pi t\|_2 - \|t\|_2\right| = O(\varepsilon).$  If $T \in T_k'$ for some $k < \ell,$ then with probability at least $1 - \frac{\delta}{n\log(2/\varepsilon)}$, $\left|\|\Pi t\|_2^2 - \|t\|_2^2\right| = \left|\|\Pi t\|_2 + \|t\|_2\right| \cdot \left|\|\Pi t\|_2 - \|t\|_2\right| = O(\varepsilon \cdot 2^{-k})$.  This means that since $\left|\|\Pi t\|_2 + \|t\|_2\right| \ge 2^{-(k+1)}$, we must have that $\left|\|\Pi t\|_2 - \|t\|_2\right| = O(\varepsilon).$  
    
    Therefore, by union bounding over all $0 \le k \le \ell$, with probability at least
\[1 - (\ell+1) \frac{\delta}{n \log (2/\varepsilon)} = 1 - \frac{(\ell+1)\delta}{n (\ell+1)} = 1 - \frac{\delta}{n},\]
    for all $t \in T,$ $\left|\|\Pi t\|_2 - \|t\|_2\right| = O(\varepsilon).$  Thus, we are done.
\end{proof}

We therefore have the following immediate corollary.

\begin{corollary} \label{MainL1Result}
    For $1 \le \varepsilon^{-2} < n$ and for any $X = \{x_1, ..., x_n\} \subset S^{d-1}$, with probability at least $1 - poly(n)^{-1}$, a randomly chosen $\Pi$ with $m = \Omega(\varepsilon^{-2} \log n)$ provides $\varepsilon$-convex hull distortion for $X$.
\end{corollary}

\subsection{Completion of the Terminal Embedding} \label{TerminalJL}

We note that the methods for completing the terminal embedding in this section are very similar to those in \cite{MahabadiMMR18}. Specifically, our proofs for Lemmas \ref{ReductionLemma1} and \ref{ReductionLemma2} are based on the proofs of \cite[Lemma 3.1]{MahabadiMMR18} and \cite[Theorem 1.5]{MahabadiMMR18}, respectively.

\begin{lemma} \label{ReductionLemma1}
    Let $x_1, ..., x_n$ be nonzero points in $\BR^d$ and let $v_i = \frac{x_i}{\|x_i\|_2}.$  Suppose that $\Pi$ provides $\varepsilon$-convex hull distortion for $V = \{v_1, -v_1, ..., v_n, -v_n\}.$  Then, for any $u \in \BR^d$, there exists a $u' \in \BR^m$ such that $\|u'\|_2 \le \|u\|_2$ and $|\langle u', \Pi x_i \rangle - \langle u, x_i\rangle| \le \varepsilon \|u\|_2 \cdot \|x_i\|_2$ for every $x_i.$
\end{lemma}

\begin{proof}
    This statement is trivial for for $u = 0,$ so assume $u \neq 0.$  It suffices to show that there always exists a $u'$ such that $\|u'\|_2 \le \|u\|_2$ and $|\langle u', \Pi v_i\rangle - \langle u, v_i \rangle| \le \varepsilon \|u\|_2$ for all $v_i.$  We will have that $|\langle u', \Pi x_i \rangle - \langle u, x_i \rangle| \le \varepsilon \|u\|_2 \cdot \|x_i\|_2$ by scaling.  Now, let $B$ be the ball in $\BR^m$ of radius $\|u\|_2$ and $\Lambda$ be the the unit $\ell_1$-ball in $\BR^n,$ where we write $\lambda \in \Lambda$ as $(\lambda_1, \dots, \lambda_n).$  Furthermore, define for $u' \in B, \lambda \in \Lambda,$
\[\Phi(u', \lambda) := \sum\limits_{i = 1}^{n} \left(\lambda_i (\langle u, v_i\rangle - \langle u', \Pi v_i\rangle) - \varepsilon |\lambda_i| \cdot \|u\|_2\right).\]
    We wish to show that there exists a $u' \in B$ such that for all $\lambda \in \Lambda,$ $\Phi(u', \lambda) \le 0.$  This clearly suffices by looking at $\lambda = \pm v_j$ for $1 \le j \le n.$
    
    Note that $\Phi$ is linear in $u'$ and concave in $\lambda,$ and that $B, \Lambda$ are compact and convex. Then, by the von Neumann minimax theorem \cite{vonNeumann28},
\[\min\limits_{u' \in B} \max\limits_{\lambda \in \Lambda} \Phi(u', \lambda) = \max\limits_{\lambda \in \Lambda} \min\limits_{u' \in B} \Phi(u', \lambda).\]
    Thus it suffices to show the right hand side is nonpositive, i.e.\ for any $u, \lambda$, there exists $u' \in B$ s.t.\ $\Phi(u', \lambda) \le 0.$  Defining $P = \sum \lambda_i v_i$, then for $u' = \|u\|_2 \cdot \frac{\Pi P}{\|\Pi P\|_2},$ it suffices to show for all $u, \lambda$ 
\[\langle u, P\rangle - \langle u', \Pi P\rangle = \langle u, P\rangle - \|u\|_2 \cdot \|\Pi P\|_2 \le \varepsilon \|\lambda\|_1 \|u\|_2.\]
    But in fact $\langle u, P \rangle \le \|u\|_2 \cdot \|P\|_2,$ so having that $\|P\|_2 - \|\Pi P\|_2 \le \varepsilon \|\lambda\|_1$ for all $\lambda \in \Lambda$ is sufficient.  For $\|\lambda\|_1 = 1$ this follows from $\Pi$ providing $\varepsilon$-convex hull distortion for $V$, and for $\|\lambda\|_1 < 1$, it follows because we can scale $\lambda$ so that $\|\lambda\|_1 = 1.$
\end{proof}

\begin{lemma} \label{ReductionLemma2}
    Let $x_1, ..., x_n \in \BR^d$ be distinct. Let $Y = \left\{\frac{x_i-x_j}{\|x_i-x_j\|_2}: i \neq j \right\}$. Moreover, suppose that $\Pi\in\BR^{m\times d}$ provides $\varepsilon$-convex hull distortion for $Y$. Then, for any $u \in \BR^d$, there exists an outer extension $f: \{x_1, \dots, x_n, u\} \to \BR^{m+1}$ with distortion $1 + O(\varepsilon),$ where $f(x_i) = \Pi x_i$.
\end{lemma}

\begin{proof}
    Note that the map $x\mapsto \Pi x$ yields a $1 + \varepsilon$-distortion embedding of $x_1, \dots, x_n$ into $\BR^m$ as it approximately preserves the norm of all points in $Y$. Therefore, we just have to verify that there exists $f(u) \in \BR^{m+1}$ such that $\|f(u) - \Pi x_i\|_2 = (1 \pm O(\varepsilon)) \|u-x_i\|_2$ for all $i$ and any $u \in \BR^d$. Fix $u \in \BR^d,$ and let $x_k$ be the closest point to $u$ among $\{x_1, \dots, x_n\}.$  By Lemma \ref{ReductionLemma1}, there exists $u' \in \BR^{m}$ such that $\|u'\|_2 \le \|u - x_k\|_2$ and for all $i$,
\begin{equation*}
    |\langle u', \Pi (x_i-x_k)\rangle - \langle u - x_k, x_i-x_k\rangle| \le \varepsilon \|u-x_k\|_2 \|x_i-x_k\|_2.
\end{equation*}

    Next, let $f(u) \in \BR^{m+1}$ be the point $(\Pi x_k + u', \sqrt{\|u-x_k\|_2^2 - \|u'\|_2^2})$. If $w_i := x_i-x_k,$ then
\begin{equation}
\|f(u) - f(x_i)\|_2^2 = \|u-x_k\|_2^2 - \|u'\|_2^2 + \|u' - \Pi w_i\|_2^2 = \|u-x_k\|_2^2 + \|\Pi w_i\|_2^2 - 2\langle u', \Pi w_i\rangle \label{eqn:bash1}
\end{equation}
    and
\begin{equation}
    \|u-x_i\|_2^2 = \|u-x_k\|_2^2 + \|w_i\|_2^2 - 2\langle u-x_k, w_i\rangle. \label{eqn:bash2}
\end{equation}

Since $\|u-x_i\|_2^2 \ge \|u-x_k\|_2^2$ and $\|u-x_i\|_2^2 \ge (\|w_i\|_2 - \|u-x_k\|_2)^2,$ we have that $\|u-x_i\|_2^2 \ge (\|u-x_k\|_2^2 + \|w_i\|_2^2)/5.$ This follows from the fact that $\max(1, (x-1)^2) \ge (x^2+1)/5$ for all $x \ge 0.$  Since $\|\Pi(x_i-x_k)\|_2 = (1 \pm \varepsilon) \|x_i-x_k\|_2$, we also have that 
\begin{equation*}
    \left|\|\Pi w_i\|_2^2 - \|w_i\|_2^2\right| = \left|\|\Pi(x_i-x_k)\|_2^2 - \|x_i-x_k\|_2^2\right| \le 3 \varepsilon \|x_i-x_k\|_2^2 = 3 \varepsilon \|w_i\|_2^2
\end{equation*}
for all $i, j,$ assuming $\varepsilon \le 1.$  Therefore, by subtracting Eq.~\eqref{eqn:bash2} from Eq.\eqref{eqn:bash1}, we have that
\begin{equation*}
\left|\|f(u) - f(x_i)\|_2^2 - \|u-x_i\|_2^2\right| \le 3 \varepsilon \|w_i\|_2^2 + 2\left|\langle u', \Pi w_i\rangle - \langle u-x_k, w_i\rangle\right| 
\end{equation*}
\begin{equation*}
\le 3\varepsilon \|w_i\|_2^2 + 2 \varepsilon \|u-x_k\|_2 \|w_i\|_2 \le 4\varepsilon\left(\|w_i\|_2^2 + \|u-x_k\|_2^2\right) \le 20 \varepsilon \|u-x_i\|_2^2, 
\end{equation*}
as desired.
\end{proof}

We summarize the previous results, which allows us to prove our main result.

\begin{theorem} \label{MainThm}
    For all $\varepsilon < 1$ and $x_1, \dots, x_n \in \BR^d,$ there exists $m = O(\varepsilon^{-2} \log n)$ and a (nonlinear) map $f: \BR^d \to \BR^m$ such that for all $1 \le i \le n$ and $u \in \BR^d,$ 
\[(1 - O(\varepsilon)) \|u-x_i\|_2 \le \|f(u)-f(x_i)\|_2 \le (1 + O(\varepsilon)) \|u
-x_i\|_2.\]
\end{theorem}

\begin{proof}
   By Corollary~\ref{MainL1Result}, there exists a $\Pi \in \BR^{m \times d}$ with $m = \Theta(\varepsilon^{-2} \log n)$ that provides $\varepsilon$-convex hull distortion for $Y$, where $Y$ is defined in Lemma \ref{ReductionLemma2}. By Lemma \ref{ReductionLemma2}, for each $u\in\BR^d$ there exists an outer extension $f^{(u)}: \{x_1, \dots, x_n, u\} \to \BR^{m+1}$ with distortion $1+O(\varepsilon)$ sending $x_i$ to $(\Pi x_i, 0)$. Therefore, we map $x_i \mapsto (\Pi x_i, 0)$ and map each $u \not\in \{x_1, ..., x_n\}$ to $f^{(u)}(u)$, which gives us a terminal embedding to $m = O(\varepsilon^{-2} \log n)$ dimensions with distortion $1+O(\varepsilon).$
\end{proof}

\subsection{Algorithm to Construct Terminal Embedding}

We briefly note how one can produce a terminal embedding into $\BR^{m}$ with a Monte Carlo randomized polynomial time algorithm, where $m = O(\varepsilon^{-2} \log n).$ By choosing a random $\Pi$ from Subsection \ref{L1JL}, we get with at least $1 - n^{-\Theta(1)}$ probability a matrix providing $\varepsilon$-convex hull distortion for our set $Y$ in Lemma \ref{ReductionLemma2}. To map any point in $\BR^d$ into $\BR^{m+1}$ dimensions, for any $u \in \BR^d$, it suffices to find a $u' \in \BR^m$ such that if $x_k$ is the point in $X$ closest to $u$, $||u'||_2 \le ||u - x_k||_2$ and for all $i$,
\[|\langle u', \Pi (x_i-x_k)\rangle - \langle u - x_k, x_i-x_k\rangle| \le \varepsilon \|u-x_k\|_2 \|x_i-x_k\|_2.\]
Assuming that $\Pi$ provides an $\varepsilon$-convex hull distortion for $Y$, such a $u'$ exists for all $u$, which means that $u'$ can be found with semidefinite programming in polynomial time, as noted in \cite{MahabadiMMR18}.

\end{document}